\documentclass[final,5p,times,twocolumn]{elsarticle}

\usepackage{amssymb}
 \usepackage{amsmath} 
\usepackage{amsthm}
\usepackage{mathrsfs,amsfonts,dsfont}
\usepackage{braket}
\usepackage{enumerate}
\usepackage{graphicx}
\usepackage{array}
\newtheorem{theorem}{Theorem}

\usepackage[hyperindex,breaklinks,bookmarks=false]{hyperref}

\journal{Optics and Laser Technology}

\begin{document}

\begin{frontmatter}

\title{Verifiable cloud-based variational quantum algorithms}

\author[inst1]{Junhong Yang}
\author[inst1]{Banghai Wang}
\ead{bhwang@gdut.edu.cn}
\author[inst2]{Junyu Quan}
\author[inst3]{Qin Li}
\ead{liqin@xtu.edu.cn}

\address[inst1]{School of Computer Science and Technology, Guangdong University of Technology, Guangzhou 510006, China}

\address[inst2]{School of Mathematics and Computational Science, Xiangtan University, Xiangtan 411105, China}

\address[inst3]{School of Computer Science, Xiangtan
University, Xiangtan 411105, China}

\begin{abstract}
Variational quantum algorithms (VQAs) have shown potential for quantum advantage with noisy intermediate-scale quantum (NISQ) devices for quantum machine learning (QML). However, given the high cost and limited availability of quantum resources, delegating VQAs via cloud networks is a more practical solution for a client with limited quantum capabilities. Recently, Shingu et al. proposed a variational secure cloud quantum computing protocol that leverages ancilla-driven quantum computation (ADQC) to perform cloud-based VQAs with minimal quantum resource consumption. However, their protocol lacks verifiability, which exposes it to potential malicious behaviors by the server. Additionally, channel loss requires frequent re-delegation as the size of the delegated variational circuit grows, complicating verification due to increased circuit complexity. This paper introduces a novel protocol that addresses these challenges by incorporating verifiability and increasing tolerance to channel loss while maintaining low quantum resource consumption for the server and requiring minimal quantum capabilities from the client.

\end{abstract}


\begin{keyword}
Cloud-based variational quantum algorithms \sep Ancilla-driven quantum computation \sep Verifiability
\end{keyword}

\end{frontmatter}

\section{Introduction}
Quantum computation~\cite{divincenzo1995quantum} has rapidly transitioned from theoretical speculation to practical application, leveraging the principles of quantum mechanics to tackle problems intractable for classical computers. Despite this progress, quantum resources remain scarce and costly, primarily accessible to large corporations. This limitation has spurred efforts to make quantum computation more accessible, especially for a client with limited quantum capabilities.

Blind quantum computation (BQC), a subset of delegated quantum computation (DQC), was first introduced by Childs~\cite{childs2005secure}. It employs quantum one-time padding~\cite{deng2004secure} in gate-based quantum computation (GBQC), enabling a client to delegate quantum computations while ensuring blindness, i.e., the client's input, output, and algorithm remain hidden from the server. Building on this concept, Broadbent et al.~\cite{broadbent2009universal} proposed the universal blind quantum computation (UBQC) protocol, also known as the BFK protocol. This protocol uses brickwork states as resource states within measurement-based quantum computation (MBQC)~\cite{briegel2009measurement} on the server side, requiring the client only to prepare a set of qubits \(\left\{ \frac{1}{\sqrt{2} }(\ket{0}+ exp(i\frac{k\pi}{4})\ket{1})\mid k \in \{0, 1, \ldots, 7\} \right\}\) while maintaining blindness. This protocol has spurred research in areas such as verification~\cite{morimae2014verification, fitzsimons2017unconditionally, barz2013experimental}, the reduction of the client's quantum capabilities ~\cite{morimae2013blind, li2021blind, morimae2013secure, li2014triple, xu2021universal, zhang2022measurement, cao2023multi}, joint computational tasks~\cite{sciarrino2023multi}, and various applications, including Shor's algorithm~\cite{das2022blind} and Grover's algorithm~\cite{gustiani2021blind}. Experimental demonstrations of BQC have also been conducted~\cite{barz2012demonstration, huang2017experimental}.

In parallel, variational quantum algorithms (VQAs)\cite{cerezo2021variational} have emerged as a promising framework for leveraging noisy intermediate-scale quantum (NISQ) devices in quantum machine learning (QML)\cite{biamonte2017quantum}. VQAs have demonstrated potential quantum advantages~\cite{khan2020machine} in domains such as quantum federated learning (QFL)\cite{chen2021federated}, quantum support vector machines (QSVMs)\cite{PhysRevLett.113.130503}, and quantum reinforcement learning (QRL)~\cite{dong2008quantum}.

Integrating VQAs with BQC provides a promising method for the client with limited quantum capabilities to delegate VQAs to a remote server via cloud networks securely. In previous work, Li et al.~\cite{li2021quantum} combined the BFK protocol~\cite{broadbent2009universal} with VQAs to implement delegated QFL. However, this approach requires significant quantum resources, with the server needing to entangle \(w \cdot d\) qubits, where \(w\) is the number of qubits required by the original NISQ algorithms and \(d\) is the depth of the brickwork states~\cite{broadbent2009universal} in the BFK protocol. Wang et al.~\cite{wang2022delegated} mitigated this by employing qubit reuse~\cite{houshmand2018minimal}, lowering the server's quantum resource consumption to \(2w+1\) qubits.

Shingu et al.\cite{shingu2022variational} further minimized the server's quantum resource consumption while upholding the principles of BQC by proposing a variational secure cloud quantum computing protocol. This protocol leverages ancilla-driven quantum computation (ADQC)\cite{anders2010ancilla} and the no-signaling principle~\cite{morimae2013blind} to implement variational quantum algorithms (VQAs) securely. Their approach requires the server to use only \(w+1\) qubits per operation. However, this protocol lacks verification and is vulnerable to potential malicious operations by the server. Additionally, it is not robust against channel loss, requiring frequent re-delegation as the size of the delegated variational circuit increases, complicating verification. This paper introduces a new protocol that extends their work by incorporating verifiability and increasing tolerance to channel loss while maintaining low quantum resource consumption for the server and requiring minimal quantum capabilities from the client.

The paper is organized as follows: Section~\ref{preliminaries} covers the preliminaries, including ADQC, VQAs, and a review of Shingu et al.'s protocol~\cite{shingu2022variational}. Section~\ref{proposed} describes the proposed protocol. Section~\ref{analysis} provides an analysis of the protocol, with a focus on verifiability, blindness, correctness, and comparisons with existing protocols. Finally, Section~\ref{conclusion} concludes the paper by discussing potential extensions and suggesting directions for future research.

\begin{figure}[t]
    \centering
    \includegraphics[width=0.45\textwidth]{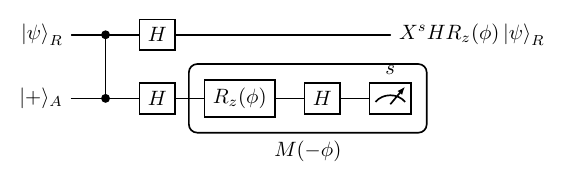}
    \caption{Circuit for the \(J(\phi)\) operator: The prepared qubits are \(\ket{\psi}_R\) and \(\ket{+}_A\), where the subscripts \(R\) and \(A\) denote the register qubits and ancillary qubit, respectively. The highlighted section represents the measurement of the ancillary qubit \(\ket{+}_A\) in the basis \(\left\{\frac{1}{\sqrt{2}}(\ket{0} \pm exp(-i\phi)\ket{1})\right\} \). After measurement, the operation \(X^{s}HR_Z(\phi)\) is obtained, where \(X\) is the Pauli \(X\) operator and \(s\) is the measurement result.}
    \label{fig:HRZ}
\end{figure}

\section{Preliminaries}
\label{preliminaries}
\subsection{Ancilla-Driven Quantum Computation}
Ancilla-driven quantum computation (ADQC)~\cite{anders2010ancilla} is a hybrid model that integrates elements of MBQC and GBQC. In ADQC, an ancillary qubit \(\ket{+}\) is coupled with a register qubit to implement a single-qubit operation \(J(\phi)\):

\begin{equation}
    J(\phi) = H R_Z(\phi),
\end{equation}

where \(\phi\) represents the designated rotation angle. Alternatively, it can be coupled with two register qubits to implement a controlled-Z gate using the fixed coupling operation \((H_R \otimes H_A)CZ_{RA}\). Here, \(CZ_{RA}\) denotes the controlled-Z gate between the register and ancillary qubits, with \(H_R\) (\(H_A\)) representing the Hadamard gate applied to the register (ancillary) qubit. The ancillary qubit is then measured in a specific basis, consistent with MBQC, to ensure determinism, and the measurement outcome determines the evolution of the register qubit(s), as illustrated in Fig.~\ref{fig:HRZ} for \(J(\phi)\) operator.

Ancillary qubits can be realized as optical photons in optical systems~\cite{agarwal2012quantum} and transmitted to distant locations after being coupled to the register qubits, enabling ADQC to be performed remotely by measuring the ancillary qubits in the basis:
\begin{equation}
    M(-\phi)=\left\{\frac{1}{\sqrt{2}}(\ket{0} \pm exp(-i\phi)\ket{1})\right\}.
\end{equation}

\subsection{Variational quantum algorithms }

Variational quantum algorithms (VQAs) utilize parameterized quantum gates, such as \(R_X\), \(R_Y\), and \(R_Z\)~\cite{cerezo2021variational}, which rotate qubits around the x, y, and z axes of the Bloch sphere, respectively. The rotation angles in the variational circuit \(\ket{\psi(\vec{\theta})}\) are optimized, where \(\ket{\psi(\vec{\theta})}\) denotes a parameterized quantum state.

When incorporating classical data into VQAs, encoding techniques such as amplitude encoding~\cite{nakaji2022approximate} are used to transform the data into quantum states. The general structure of the circuit is illustrated in Fig.~\ref{fig:VQC}. VQA circuits can be conceptualized as quantum neural networks~\cite{jeswal2019recent}, where qubits serve as nodes, and quantum gates, represented by matrices, correspond to the weights of a neural network, directly influencing the qubit states.
\begin{figure}[t]
    \centering
    \includegraphics[width=0.35\textwidth]{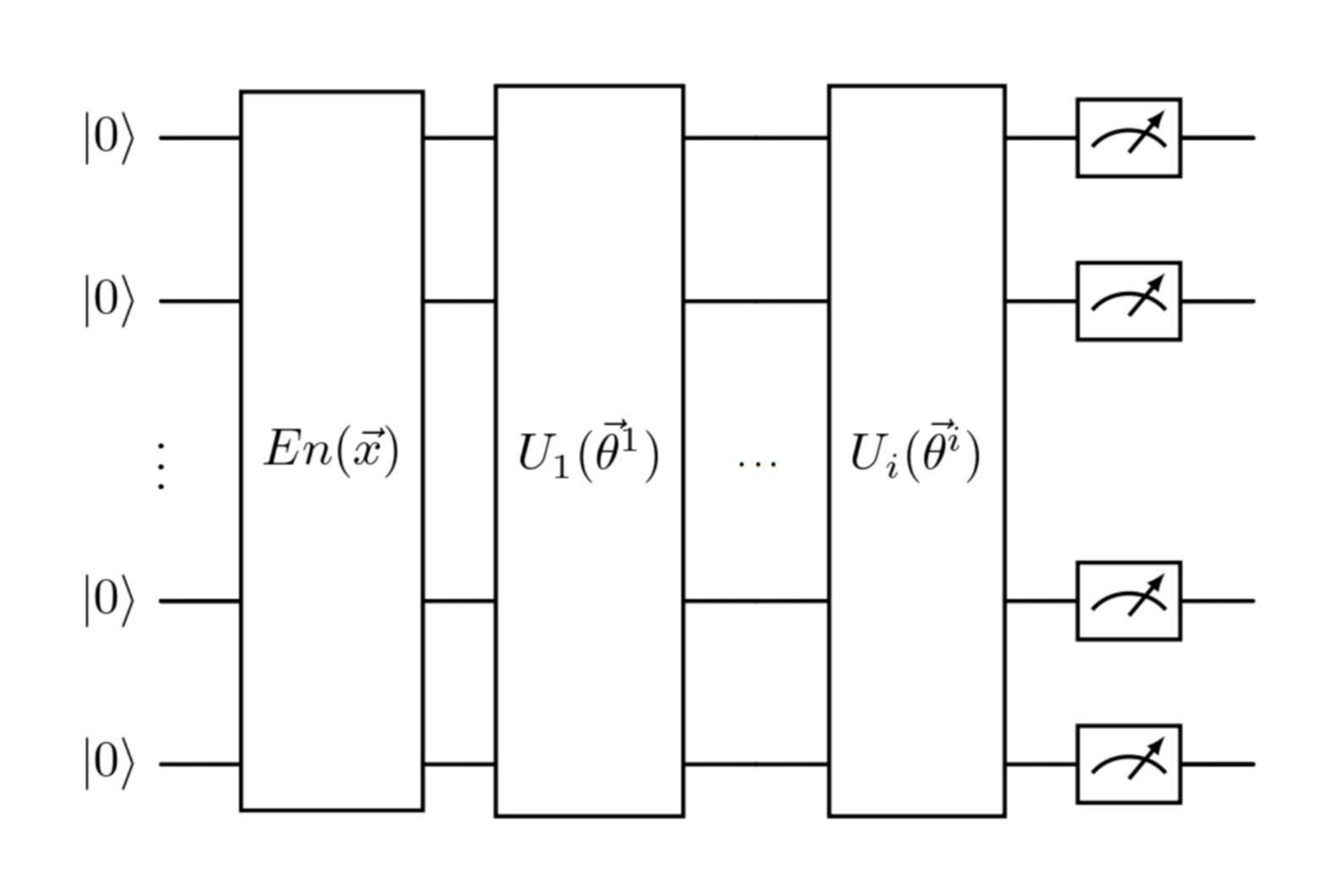}
    \caption{General variational quantum circuit: All qubits are initialized in the state \(\ket{0}\). The operator \(En\) encodes classical data \(\vec{x}\) into the quantum state \(En(\vec{x})\ket{0}^{\otimes w}\). The unitary operator \(U(\vec{\theta}) = \prod_{i=1}^{n} U_i(\vec{\theta}^i)\) represents the variational layers, forming a specifically designed ansatz, where each \(U_i(\vec{\theta^i})\) corresponds to the \(i\)-th layer of the \(n\) variational layers.}
    \label{fig:VQC}
\end{figure}
The variational parameters \(\vec{\theta} = \left\{ \theta_1, \theta_2, \ldots, \theta_L \right\}\), where \(L\) denotes the total number of parameters, are iteratively updated and optimized by a classical optimizer, such as Adam~\cite{KingBa15}, using measurement outcomes from the quantum circuit along with the cost function:
\begin{equation}
\begin{aligned}
    C(\vec{\theta}) &= f(E(\vec{\theta})) \\
              &= f(\bra{\psi(\vec{\theta})}O\ket{\psi(\vec{\theta})}),
\end{aligned}
\end{equation}
where \(E(\vec{\theta})\) denotes the expectation value of the output state \(\ket{\psi(\vec{\theta})}\) for a given set of parameters \(\vec{\theta}\), and \(O\) is the measurement observable. The function \(f\) represents the objective function, such as mean squared error (MSE) or cross-entropy.

The parameters \(\vec{\theta} = \left\{\theta_j\right\}_{j=1}^L\) are iteratively updated via gradient descent, such that \(\vec{\theta} \gets \vec{\theta} - \eta \nabla C(\vec{\theta})\), where \(\eta\) denotes the learning rate. The gradients \(\nabla C(\vec{\theta}) = \left\{\frac{\partial C(\vec{\theta}) }{\partial \theta_j} \right\}_{j=1}^L\) are computed using the parameter-shift rule~\cite{schuld2019evaluating}. Specifically, the gradient \(\frac{\partial C(\vec{\theta}) }{\partial \theta_j}\) is given by:

\begin{equation}
\begin{aligned}
    \frac{\partial C(\vec{\theta}) }{\partial \theta_j}
    &= \frac{\partial f(E(\vec{\theta}))}{\partial E(\vec{\theta})} \cdot \frac{\partial E(\vec{\theta})}{\partial \theta_j}\\
    &= \frac{1}{2} \frac{\partial f(E(\vec{\theta}))}{\partial E(\vec{\theta})} \cdot\left[E\left(\theta_{j+}\right) - E\left(\theta_{j-} \right) \right],
\end{aligned}
\label{gradient}
\end{equation}
where \(\theta_{j\pm} = \left\{ \theta_1, \ldots, \theta_j \pm \frac{\pi}{2}, \ldots, \theta_L \right\}\).

\subsection{Review of Shingu et al.'s Protocol}
The variational secure cloud quantum computing protocol proposed by Shingu et al.~employs the circuit depicted in Fig.~\ref{fig:HRZ} to perform single-qubit gates and necessitates additional gates to execute two-qubit controlled gates, as shown in Fig.~\ref{fig:TCG}.

\begin{figure*}[t]
    \centering
    \includegraphics[width=0.65\textwidth]{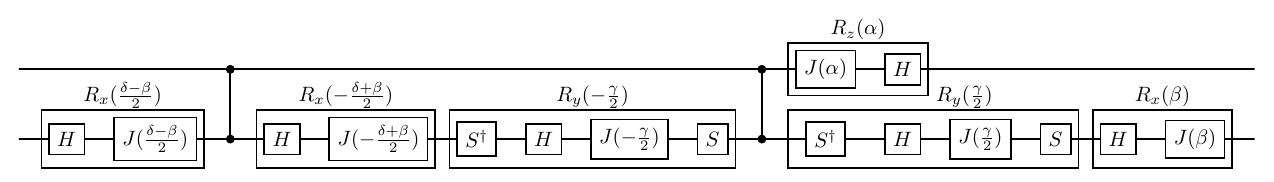}
    \caption{Circuit for the two-qubit controlled gates: The server interacts with the client to obtain 6 \(J(\phi)\) operators with parameters \(\alpha\), \(\beta\), \(\gamma\), and \(\delta\).}
    \label{fig:TCG}
\end{figure*}

Specifically, a single-qubit gate \(U\) is implemented using three consecutive \(J(\phi)\) operators with \(\phi_1\), \(\phi_2\), and \(\phi_3\), following Euler's rotation theorem~\cite{gothen2023euler}. Here, \(U = J(\phi_1)J(\phi_2)J(\phi_3)\) and \(HU = R_Z(\phi_1)R_X(\phi_2)R_Z(\phi_3)\), where \(\phi_1\), \(\phi_2\), and \(\phi_3\) are Euler angles. In contrast, a two-qubit controlled gate \(CU\) is realized by decomposing it into multiple \(J(\phi)\) operators to be delegated, along with additional \(CZ\) gates, Hadamard (\(H\)) gates, the \(S\) gate, and its conjugate transpose, the \(S^\dagger\) gate, which are directly performed by the server, as shown in Fig.~\ref{fig:TCG}. The client receives and measures ancillary qubits sent by the server to evolve the register qubits according to ADQC. Subsequently, the server measures the output register qubits and sends the results to the client.

The protocol proceeds as follows:

{\bf A1(Preparation phase):}
\begin{enumerate}
    \item[\bf A1-1:] The server publicly announces the following: the number of original and variational quantum circuits for gradient calculations, \(G\); the set of unitary operators \(\{U^{(c)}_{AN}\}_{c=1}^{G}\), where \(c\) denotes the \(c\)-th circuit; the set of measurement observables \(\{\hat{A}_{1}^{(c)}, \hat{A}_{2}^{(c)}, \ldots, \hat{A}_{K^{(c)}}^{(c)}\}_{c=1}^{G}\), where \(K^{(c)}\) is the number of observables measured in the \(c\)-th circuit; the number of circuit repetitions \(\{\mathcal{R}^{(c)}\}_{c=1}^{G}\); the initial states \(\{\ket{\psi_{out}^{(c)}(\vec{\theta}[0])}\}_{c=1}^{G}\), where \(\vec{\theta}[0]\) indicates the initial parameters at the current iteration step; the number of variational parameters, \(L\); and the total number of iteration steps for VQAs, \(\mathcal{I}\).
    \item[\bf A1-2:] The server prepares \(w\) register qubits \(\ket{0}_R\) and one ancillary qubit \(\ket{+}_A\).
\end{enumerate}

{\bf A2(Computation phase):}
\begin{enumerate}
    \item[\bf A2-1:] Adopt the quantum circuits \(\{U^{(c)}_{\rm{AN}}\}_{c=1}^{G}\) using the circuits depicted in Fig.~\ref{fig:HRZ} and Fig.~\ref{fig:TCG} to generate the trial wave functions \(\left\{\left|\psi^{(c)}(\vec{\theta}[1])\right\rangle\right\}_{c=1}^{G}\).
    \item[\bf A2-2:] The server measures the output register qubits' states with the measurement observables \(\{\hat{A}^{(c)}_1, \hat{A}^{(c)}_2, \ldots, \hat{A}^{(c)}_{K^{(c)}}\}_{c=1}^{G}\) and sends the results to the client via classical communication.
    \item[\bf A2-3:] The client compensates for the Pauli byproduct effect.
    \item[\bf A2-4:] The server reprepares ancillary qubit \(\ket{+}_A\).
\end{enumerate}

{\bf A3(Parameters updating phase):}
\begin{enumerate}
    \item[\bf A3-1:] The server and the client repeat  \textbf{A2} \(\{\mathcal{R}^{(c)}\}_{c=1}^{G}\) times to derive the expectation values of \(\{\hat{A}^{(c)}_1, \hat{A}^{(c)}_2, \ldots, \hat{A}^{(c)}_{K^{(c)}}\}_{c=1}^{G}\). 
    \item[\bf A3-2:] The client updates the parameters using its optimizer according to Eq.~\ref{gradient}, resulting in \(\vec{\theta}[2] = (\theta_1[2], \ldots, \theta_L[2])^T\) for the next iteration.
    \item[\bf A3-3:] The client computes the measurement angles for the circuits \(\{U^{(c)}_{\rm{AN}}\}_{c=1}^{G}\) using the parameters \(\vec{\theta}[2]\). 
    \item[\bf A3-4:] The client and server reiterate the above steps \((\mathcal{I}-2)\) times with \(\{U^{(c)}_{\rm{AN}}\}_{c=1}^{G}\) and \(\vec{\theta}[j]\). Based on the results from the \(j\)-th step, the client's optimizer updates the parameters to \(\vec{\theta}[j+1]\) for \(j = 2, 3, \ldots, \mathcal{I}-1\).

\end{enumerate}

This protocol enables the delegation of VQAs through interaction between the server and the client, requiring only \(w\) register qubits and a single ancillary qubit. However, if any of the coupled ancillary qubits are lost during transmission, the circuit needs to be re-delegated, which results in a low tolerance to channel loss. Even the most efficient single-photon detectors in optics, which have shown 99\% efficiency in recent studies~\cite{kuzanyan2018three, wenner2014catching}, are not immune to this issue. For sufficiently large circuits, the probability of re-delegation becomes significant. For example, consider a circuit where each block consists of an average of four single-qubit gates and one two-qubit control gate. A circuit with six blocks would require \(4 \times 3 \times 6 + 1 \times 6 \times 6 = 108\) \(J(\phi)\) operators. The probability of needing to re-delegate the circuit is approximately \((1 - 0.99^{108}) \approx 66.22\%\) per delegated \(J(\phi)\) operator, increasing exponentially with the number of \(J(\phi)\) operators.

Due to channel loss, adding dummy gates for verification becomes impractical, as the circuit would require enough extra trap wires for dummy gates without affecting the wires used for computation. This increases the circuit size and leads to frequent circuit re-delegation. One potential solution to mitigate channel loss is to share Bell pairs between the server and the client, with repeated entanglement generation until successful~\cite{morimae2013blind}. However, during each delegation of the \(J(\phi)\) operator, as illustrated in Fig.~\ref{fig:HRZ}, the ancillary qubit is sent to the client for measurement after being coupled with the register qubit. Although the server can send one half of the Bell pair to the client for measurement, additional measurements on the remaining half are required on the server side after being coupled with the register qubit. However, the server cannot perform these measurements in the context of Fig.~\ref{fig:HRZ}, where measurement on the ancillary qubit is performed by the client, rendering the use of Bell pairs ineffective in this context.

\section{The Proposed Protocol}
\label{proposed}
We adopt the strategy outlined in~\cite{sueki2013ancilla}, where the server performs measurements after coupling the ancillary qubits with the register qubits, effectively addressing the limitations of using Bell pairs to tolerate channel loss in Shingu et al.’s protocol. The client is required to perform measurements in the bases \(\left\{ M(-\frac{k\pi}{4}) \mid  k \in \{0, 1, \ldots, 7\} \right\}\). For verification, we use trap qubits to create dummy gates in the trap wires as the circuit is transformed into the universal gate patterns shown in Fig.~\ref{fig:brickwork}(a). Additionally, we delegate encrypted measurements of the output register qubits to the server, making it difficult to identify the trap wires in the circuits.
\begin{figure*}[t]
    \centering
    \includegraphics[width=0.75\textwidth]{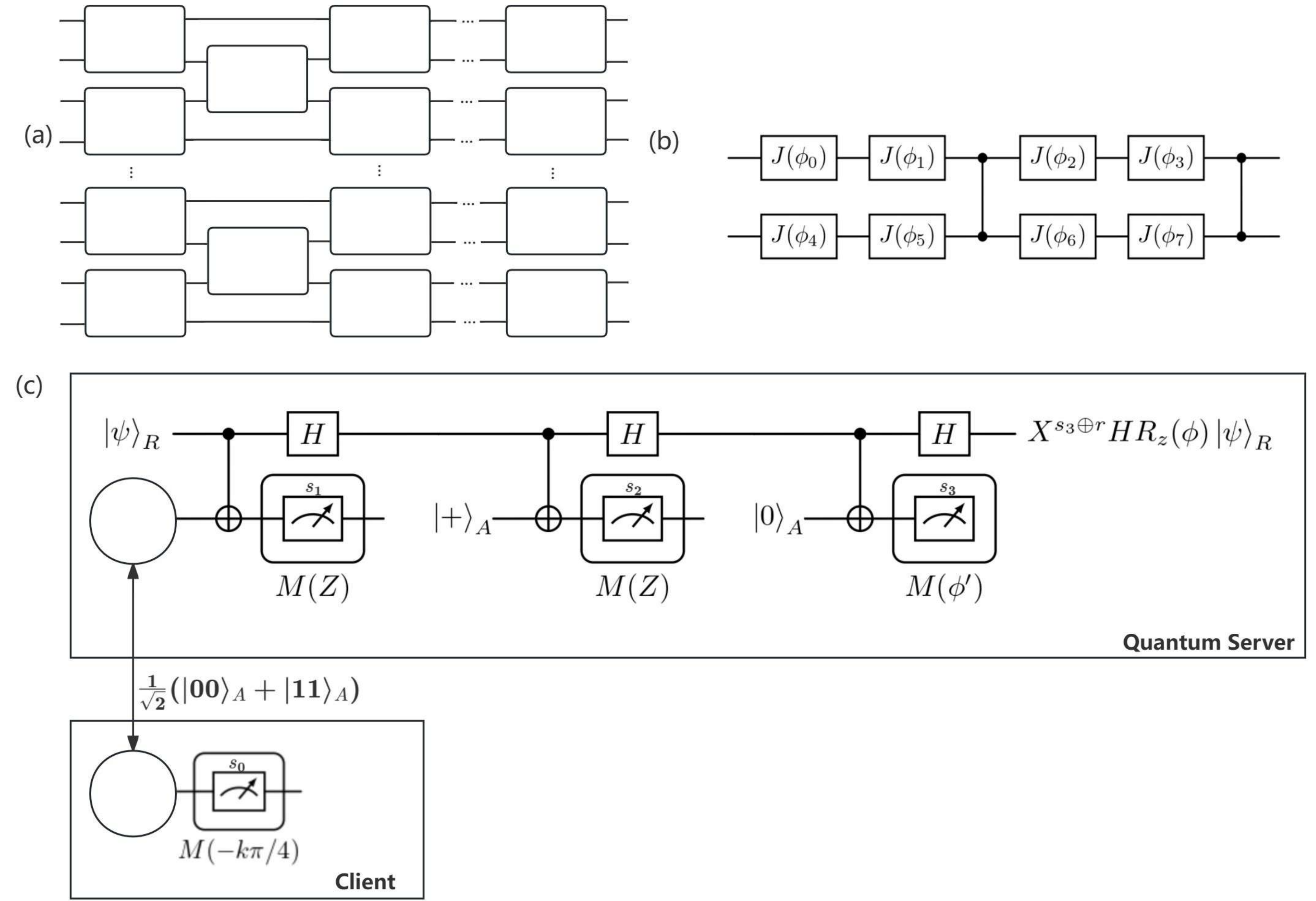}
    \caption{(a) The universal gate patterns: Composed of multiple blocks, each representing a gate pattern. (b) The gate pattern circuit: Consisting of 8 \(J(\phi)\) operators and 2 \(CZ\) gates. (c) Realization of the \(J(\phi)\) operator: The server sends one half of a Bell pair to the client, who measures in the basis \(M\left(-\frac{k\pi}{4}\right)\) with result \(s_0\). The server then performs operations on the ancillary and register qubits, including a fixed coupling operation \((H_R \otimes I_A)CX_{RA}\), where \(CX\) denotes the controlled-X gate and \(I\) is the identity gate, followed by measurements in the \(Z\) basis or in the \(M(\phi')\) basis on the ancillary qubits with results \(s_1\) to \(s_3\).}

    \label{fig:brickwork}
\end{figure*}

The client delegates the \(J(\phi)\) operators within the universal gate patterns shown in Fig.~\ref{fig:brickwork}(a), while the server provides the \(CZ\) gates. This process includes preparation, computation, verification, and parameter updating phases, as briefly illustrated in Fig.~\ref{fig:briefprocess}. The specific steps are as follows:

\begin{figure}[t]
    \centering
    \includegraphics[width=0.45\textwidth]{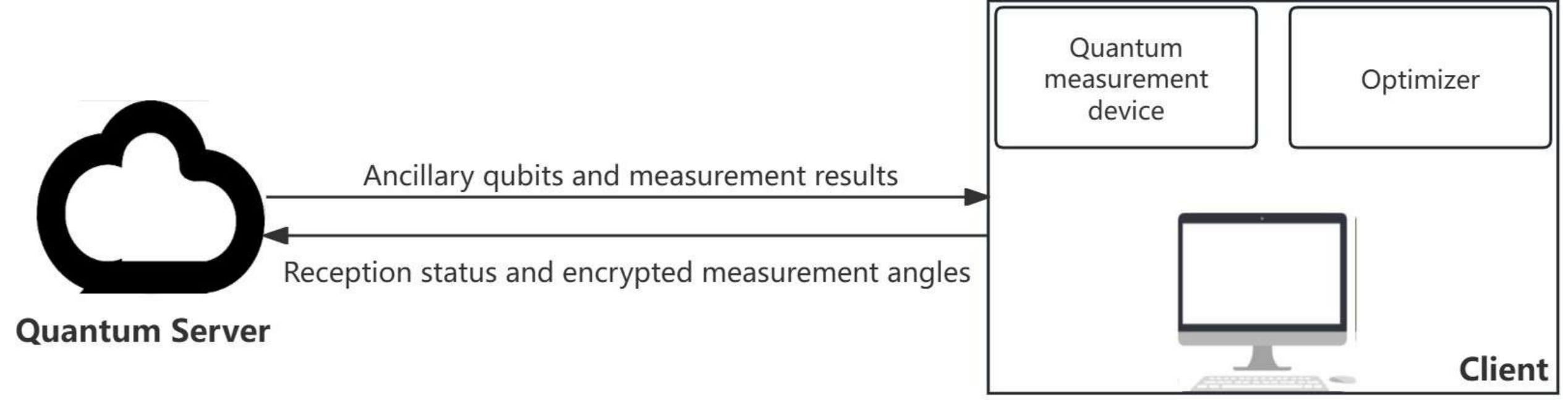}
    \caption{Brief process of the proposed protocol: The server sends an ancillary qubit, as one half of a Bell pair, to the client, who measures it in the basis \(M(-\frac{k\pi}{4})\). The client then sends the reception status back to the server, requesting a resend if the qubit is lost. The server performs encrypted measurements as instructed by the client and returns the results, enabling the client to verify the server's honesty and calculate encrypted measurement angles and gradients for the optimizer.
}
    \label{fig:briefprocess}
\end{figure}

{\bf  B1(Preparation phase):}
\begin{enumerate}
    \item[\bf B1-1:] The server publicly announces the following: the number of original and variant variational quantum circuits for gradient calculation, \(G\); the set of transformed unitary operators, \(\{U^{(c)}_{AN}\}_{c=1}^{G}\); the number of circuit repetitions, \(\{\mathcal{R}^{(c)}\}_{c=1}^{G}\); the initial states, \(\{\ket{\psi_{out}^{(c)}(\vec{\theta}[0])}\}_{c=1}^{G}\), for the \(G\) circuits; the number of variational parameters, \(L\); the total number of iteration steps for VQAs, \(\mathcal{I}\); and the size of the transformed circuits, \(\{N \times M\}_{c=1}^{G}\), where \(N=3w\) and \(M\) represent the number of input qubits and the circuit depth, respectively.
    \item[\bf B1-2:]  The server prepares \(N\) register qubits \(\ket{0}_R\), one ancillary qubit \(\ket{+}_A\), one ancillary qubit \(\ket{0}_A\) and two ancillary qubits \(\frac{1}{\sqrt{2}}(\ket{00}_A+\ket{11}_A)\) as a Bell pair. Meanwhile, the client chooses \(2N/3\) register qubits as trap qubits for verification, which is optimal~\cite{morimae2014verification}.
\end{enumerate}
{\bf B2(Computation phase):}
\begin{enumerate}
    \item[\bf B2-1:] The server sends one qubit from the Bell pair to the client, who measures it in the basis \(M\left(-\frac{k\pi}{4}\right)\), yielding the result \(s_0\), where \(k\) is randomly and uniformly selected from \(\{0, 1, \ldots, 7\}\). The client generates a random bit \(\text{recv\_status}\) to indicate the reception status, setting it to 0 if the qubit arrived successfully or to 1 if it was lost. This status is communicated to the server, who generates a new Bell pair and resends half of it to the client until the qubit arrives successfully. The other half of the Bell pair will eventually be in the state:

\begin{equation}
    Z^{s_0} R_Z\left(\frac{k\pi}{4}\right) \ket{+}_A = \frac{1}{\sqrt{2}} \left[ \ket{0}_A + \exp\left(i \left( \frac{k}{4} + s_0 \right)\pi \right) \ket{1}_A \right].
\end{equation}

   \item[\bf B2-2:] The server employs the transformed circuits \(\{U^{(c)}_{\rm{AN}}\}_{c=1}^{G}\) using the \(J(\phi)\) operator, as depicted in Fig.~\ref{fig:brickwork}(c), to generate the trial wave functions \(\left\{\left|\psi^{(c)}(\vec{\theta}[1])\right\rangle\right\}_{c=1}^{G}\). In each delegation of the \(J(\phi)\) operator, the \(Z\) measurement results \(s_1\) and \(s_2\) are first sent to the client. The client then uses \(s_1\) to compute the encrypted measurement angle:
   \begin{equation}
       \phi' = -\phi + (-1)^{s_1} \left( \frac{k}{4} + s_0 \right)\pi + r\pi,
       \label{eq6}
   \end{equation}
    where \(r \in \{0, 1\}\) is a random bit chosen by the client. This encrypted angle \(\phi'\) is sent to the server, which measures in the basis \(M(\phi')\). The resulting measurement outcome \(s_3\) is then sent to the client. The actual measurement angle \(\phi\) is adaptively chosen based on prior Pauli byproducts to ensure determinism in MBQC.
  \item[\bf B2-3:] In each transformed circuit, \(2N/3\) of the wires are randomly designated as trap wires, where dummy operations are implemented. On half of these trap wires, the identity gate is applied \(M\) times on each wire. On the remaining trap wires, a single Hadamard gate is randomly inserted on each wire, while the identity gate occupies the other \(M-1\) positions.

  \item[\bf B2-4:] The server measures the output register qubits, each with Pauli byproducts \(\{ X^{x_j} \}_{j=1}^{N}\) and \(\{ Z^{z_j} \}_{j=1}^{N}\), in the \(Z\) basis. To measure in the \(X\) basis, the client applies an additional Hadamard gate at the end of the original circuit before measuring in the \(X\) basis. This transformation can be expressed as:

\begin{equation}
    \begin{split}
        \braket{X_{R_j}} 
        &= \sum_{m=0}^{1} \left( \bra{\Psi}_{R_j} Z^{x_j} X^{z_j} \right) H \ket{m}\bra{m} H \left( X^{z_j} Z^{x_j} \ket{\Psi}_{R_j} \right) \\
        &= \sum_{m=0}^{1}  \left( \bra{\Psi}_{R_j} Z^{x_j} X^{z_j} H \right) \ket{m}\bra{m} \left( H X^{z_j} Z^{x_j} \ket{\Psi}_{R_j} \right) \\
        &= \sum_{m=0}^{1} \left( \bra{\Psi}_{R_j} H X^{x_j} Z^{z_j} \right) \ket{m}\bra{m} \left( Z^{z_j} X^{x_j} H \ket{\Psi}_{R_j} \right) \\
        &= \braket{Z_{R_j}}.
    \end{split}
\end{equation}

where \(\ket{\Psi}_{R_j}\) is the \(j\)-th output register qubit of \(\ket{\Psi}_{R}\) for \(j=1,\dots,N\). \(\braket{X_{R_j}}\) is the expectation value of \(X\) in the state \(X^{z_j} Z^{x_j} \ket{\Psi}_{R_j}\), and \(\braket{Z_{R_j}}\) is the expectation value of \(Z\) in the state \(X^{x_j} Z^{z_j} H \ket{\Psi}_{R_j}\).

\item[\bf B2-5:] The client compensates for the Pauli byproducts on all output register qubits by flipping the measurement results according to \(\left \{ x_j \right \} _{j=1}^N\) and \(\left \{ z_j \right \} _{j=1}^N\).

\item[\bf B2-6:] The server resets one ancillary qubit to \(\ket{+}_A\), one to \(\ket{0}_A\), and reprepares Bell pair \(\frac{1}{\sqrt{2}}(\ket{00}_A + \ket{11}_A)\).
\end{enumerate}

{\bf B3(Verification phase):}

\begin{enumerate}

    \item[\bf B3-1:] At the end of \textbf{B2-3}, the output register qubits consist of non-trap qubits and trap qubits in a random permutation designed by the client. The output state can be written as:
    \begin{equation}
        \ket{\Psi}_R = \sigma_{q} P\left(\ket{\psi}_{out}^{\frac{N}{3}} \otimes \ket{0}_{T}^{\frac{N}{3}} \otimes \ket{+}_{T}^{\frac{N}{3}}\right),
    \end{equation}
    where \(\ket{\psi}_{out}\) represents the non-trap qubits, and \(\sigma_{q}\) is the Pauli byproduct operator. The subscript \(R\) includes both non-trap qubits (\(out\)) and trap qubits (\(T\)), with \(P\) as the permutation.
    \item[\bf B3-2:] The client instructs the server to measure all desired trap qubits, \(\ket{0}_T^{\frac{N}{3}}\) in the \(Z\) basis and \(\ket{+}_T^{\frac{N}{3}}\) in the \(X\) basis, in \textbf{B2-4}. If any undesired output \(\ket{1}_T\) in the \(Z\) basis or \(\ket{-}_T\) in the \(X\) basis is obtained after compensating for the Pauli byproduct effect, the protocol is terminated.
\end{enumerate}

{\bf B4 (Parameters updating phase):}
\begin{enumerate}

    \item[\bf B4-1:] The server and the client repeat \textbf{B2} to \textbf{B3} \(\{\mathcal{R}^{(c)}\}_{c=1}^{G}\) times to derive the expectation values of \(Z\) or \(X\) in each output register qubit.
    \item[\bf B4-2:] The client utilizes its optimizer to update the parameters according to Eq.~\ref{gradient}, resulting in \(\vec{\theta}[2] = (\theta_1[2], \cdots, \theta_L[2])^T\) for the next step.
    \item[\bf B4-3:] The client calculates the measurement angles for circuits \(\{U^{(c)}_{\rm{AN}}\}_{c=1}^{G}\) using \(\vec{\theta}[2]\).
    \item[\bf B4-4:] The client and server reiterate the above steps \((\mathcal{I}-2)\) times with \(\{U^{(c)}_{\rm{AN}}\}_{c=1}^{G}\) and \(\vec{\theta}[j]\). Based on the results from the \(j\)-th step, the client's optimizer updates the parameters to \(\vec{\theta}[j+1]\) for \(j = 2, 3, \ldots, \mathcal{I}-1\).

\end{enumerate}

\section{Analysis}
\label{analysis}
\subsection{Verifiability}
We use trap qubits to detect malicious operations on the output register qubits. The verifiability of the proposed protocol is demonstrated in Theorem I, which follows a method similar to that in \cite{morimae2014verification}, where further details can be found.

\begin{theorem}
    The probability of the client being tricked by the server is exponentially small.
\end{theorem}

\begin{proof}
    
    If the server is malicious, it may deviate from the state: 
    \begin{equation}
        \rho = \sigma_q |\Psi\rangle_R \langle\Psi|_R\sigma^{\dagger}_q
    \end{equation}

 to an arbitrary state. However, due to the completely positive trace preserving (CPTP) map~\cite{morimae2013secure}, this deviation can be detected as a random Pauli attack.

    Suppose \(\sigma_\alpha\) represents a random Pauli attack, where the weight of \(\sigma_\alpha\) (\(|\alpha|\)) is the number of non-trivial operators in \(\sigma_\alpha\), such that \(|\alpha| = a + b + c\), where \(a\), \(b\), and \(c\) are the numbers of \(X\), \(Z\), and \(XZ\) operators in \(\sigma_\alpha\), respectively. We have \(|\alpha| = a + b + c \leq 3\max(a,b,c)\). When these operators are applied to the output trap qubits \(\ket{0}_T\) and \(\ket{+}_T\), \(X\) will only change \(\ket{0}_T\) to \(\ket{1}_T\), \(Z\) will only change \(\ket{+}_T\) to \(\ket{-}_T\), and \(XZ\) will change both \(\ket{0}_T\) and \(\ket{+}_T\) to \(\ket{1}_T\) and \(\ket{-}_T\), respectively.

    We can calculate the probabilities that each operator in \(\sigma_\alpha\) does not change any trap qubits. Suppose \(\max(a,b,c) = a\). An \(X\) operator that does not change any trap qubits will only act on \(\ket{+}_T\) and non-trap qubits, the number of which is \(2N/3\). Thus, the probability is:
    \begin{equation}
    \begin{split}
        \frac{C\left(\frac{2N}{3},a\right)}{C(N,a)} 
        &= \left(\frac{2}{3}\right)^a \frac{\prod_{k=0}^{a-1}\left(N-\frac{3}{2}k\right)}{\prod_{k=0}^{a-1}(N-k)} \\
        &\leq \left(\frac{2}{3}\right)^a \leq \left(\frac{2}{3}\right)^{\frac{|\alpha|}{3}}.
    \end{split}
\end{equation}
    Similarly, for \(\max(a,b,c) = b\), the probability \(\frac{C\left(\frac{2N}{3},b\right)}{C(N,b)} \leq \left(\frac{2}{3}\right)^{\frac{|\alpha|}{3}}\), and for \(\max(a,b,c) = c\), the probability \(\frac{C\left(\frac{N}{3},c\right)}{C(N,c)} \leq \left(\frac{1}{3}\right)^{\frac{|\alpha|}{3}}\) can be calculated in the same way.

    The proposed protocol requires the client to send qubits, qubit reception statuses, and encrypted measurement angles, which seems to conflict with the no-signaling principle. However, the permutation \(P\) is not transmitted, ensuring its secrecy. The server is only aware of measuring in the \(Z\) basis, while the \(X\) basis is implemented by delegating an additional Hadamard gate at the end of the original circuit. Consequently, the server cannot distinguish between trap and non-trap qubits, thereby further protecting \(P\) without any leakage to the server.

    After the client randomly selects a permutation \(P\), the probability that \(P^\dagger \sigma_\alpha P\) does not alter any trap qubits is at most \(\left(\frac{2}{3}\right)^{\frac{|\alpha|}{3}}\). Consequently, the probability that the server deceives the client is at most:
    \begin{equation}
        \prod_{c=1}^{G}\left(\frac{2}{3}\right)^{\frac{\mathcal{R}^{(c)}|\alpha|}{3}} = \left(\frac{2}{3}\right)^{\frac{\sum_{c=1}^{G} \mathcal{R}^{(c)} |\alpha|}{3}},
    \end{equation}
    where \(G\) is the number of original and variant quantum circuits used to calculate gradients, and \(\{\mathcal{R}^{(c)}\}_{c=1}^{G}\) is the set of circuit repetitions. This probability becomes exponentially small when \(\sum_{c=1}^{G} \mathcal{R}^{(c)}\) is sufficiently large, ensuring the protocol’s verifiability.

\end{proof}
 
\subsection{Blindness and Correctness}
We utilize universal gate patterns in the proposed protocol to ensure both blindness and correctness during computation, with minimal information leakage. Specifically, only the size of the delegated circuit, corresponding to the size of the universal gate patterns, is revealed. The following theorems establish the blindness and correctness of the proposed protocol in the context of cloud-based VQAs.

\begin{theorem}
    The proposed protocol guarantees input, output, and algorithm blindness.
\end{theorem}

\begin{proof}
    \textbf{Input Blindness:} After the client measures one qubit of the Bell pair, the remaining qubit is left in a maximally mixed state:
\begin{equation}
    \frac{1}{16} \sum_{s_0=0}^{1}\sum_{k=0}^{7} \left(Z^{s_0}R_Z\left(\frac{k\pi}{4}\right) \ket{+}_A\bra{+}_AR_Z^{\dagger}\left(\frac{k\pi}{4}\right)Z^{s_0\dagger} \right) = \frac{I}{2},
\end{equation}
with the value of \(k\) hidden from the server.

    \textbf{Output Blindness:} After applying the \(J(\phi)\) operator, the register qubit is \(X^{ s_3 \oplus r}HR_Z(\phi)\ket{\psi}_R\). Upon completion of the delegated circuit, the output register qubits are \(X^{x_{out}}Z^{z_{out}}U\ket{\Psi}_R\), where \(X^{x_{out}}Z^{z_{out}}\) are Pauli byproducts. Since the server cannot determine \(r\), it cannot compensate for \(X^{s_3 \oplus r}\), and therefore cannot compensate for \(X^{x_{out}}Z^{z_{out}}\), leaving the output qubits in a maximally mixed state.

    \textbf{Algorithm Blindness:} When the client sends the encrypted measurement angle \(\phi' = -\phi + (-1)^{s_1}\left(\frac{k\pi}{4} + s_0\pi\right) + r\pi\) to the server, the true measurement angle \(\phi\) remains concealed, as \(k\) and \(r\) are randomly chosen, and \(s_0\) is never revealed.

    Therefore, the server can only deduce the general structure of the universal gate patterns.
\end{proof}

\begin{theorem}
    The proposed protocol ensures correctness throughout the computation.
\end{theorem}

\begin{proof}
    In each gate pattern, as depicted in Fig.~\ref{fig:brickwork}(b), the server directly applies the \(CZ\) gate, while the \(J(\phi)\) operator is implemented using the circuit shown in Fig.~\ref{fig:brickwork}(c). This procedure results in the state \(HR_Z(\phi)\ket{\psi}_R\) with Pauli byproducts. Further details on this mathematical derivation are provided in ~\ref{appendix}.
   
    The client subsequently adjusts the measurement results provided by the server to account for these Pauli byproducts, ensuring the correct realization of \(J(\phi)\) operators within the universal gate patterns. Each gate pattern is parameterized by \(\{\phi_i\}_{i=0}^7\), enabling the implementation of arbitrary single-qubit gates. For example, the operation \(U_1 \otimes U_2\) can be achieved by setting \(\phi_5 = \phi_7 = 0\) while appropriately selecting the remaining parameters.  Here, \(U_1\) and \(U_2\) denote arbitrary single-qubit gates, which can be decomposed as \(U_1 = R_Z(\phi_0) R_x(\phi_1) R_Z(\phi_2)\) and \(U_2 = R_Z(\phi_4) R_x(\phi_5) R_Z(\phi_6)\), following Euler's rotation theorem~\cite{gothen2023euler}. Additionally, the \(CX\) gate can be realized within this pattern by setting \(\phi_2 = \phi_5 = \frac{\pi}{2}\) and \(\phi_7 = -\frac{\pi}{2}\), while setting the others to \(0\).

    Using these gates, any quantum gate in VQAs can be realized. Moreover, the delegated computation adheres to the rules of MBQC, where measurement angles are adaptively chosen to move all Pauli byproducts to the leftmost position, resulting in \(X^{x_{\text{out}}}Z^{z_{\text{out}}}U\ket{\Psi}_R\). By appropriately correcting the measurement results, the client ultimately obtains the desired state \(U\ket{\psi}_{\text{out}}\).
\end{proof}

 \begin{table*}
\label{tb:1}
\newcommand{\tabincell}[2]{\begin{tabular}{@{}#1@{}}#2\end{tabular}}
\caption{Comparison among different protocols for cloud-based VQAs}
\label{protocols}       
\resizebox{\textwidth}{12mm}{
\begin{tabular}{ccccc}
\hline\noalign{\smallskip}
    & Verifiable & Client's quantum capabilities  & Tolerance to channel loss & Server's quantum resource consumption \\
\hline\noalign{\smallskip}
Li et al.'s BFK-based protocol~\cite{li2021quantum} & Yes & Prepare qubits & Yes & $6w\cdot d$ \\
\hline\noalign{\smallskip}
Wang et al.'s BFK-based protocol~\cite{wang2022delegated} & Yes & Measure qubits & Yes & $6w+1$ \\
\hline\noalign{\smallskip}
Shingu et al.'s protocol~\cite{shingu2022variational} & No & Measure qubits & No & $w+1$ \\
\hline\noalign{\smallskip}
The proposed protocol & Yes & Measure qubits & Yes & $3w+4$ \\
\hline\noalign{\smallskip}
\end{tabular}}
\end{table*}
\subsection{Comparisons}

Table~\ref{protocols} compares the proposed protocol with related protocols for cloud-based VQAs. Our protocol extends the method introduced by Shingu et al.~\cite{shingu2022variational} by incorporating verification and increasing tolerance to channel loss. Regarding the client's quantum capabilities, our protocol requires only minimal quantum resource consumption; specifically, the client needs only to perform measurements in specific bases, which is significantly easier than preparing qubits in optical systems~\cite{morimae2013blind}. 

An additional advantage of our protocol is its improved tolerance to channel loss. In Shingu et al.'s protocol, the loss of any ancillary qubit during transmission necessitates the re-delegation of the entire circuit, which becomes particularly problematic for large circuits. In contrast, our protocol allows the server to perform encrypted measurements, thereby achieving tolerance by sharing Bell pairs between the server and client. While both Wang et al.'s~\cite{wang2022delegated} and Li et al.'s~\cite{li2021quantum} BFK-based protocols also tolerate channel loss, Wang et al.'s protocol requires \(6w \cdot d\) qubits with verification, and Li et al.'s protocol requires \(6w + 1\) qubits with verification, where \(w\) is the number of qubits in the original NISQ algorithms and \(d\) is the depth of the brickwork states in the BFK protocols. In contrast, our protocol, adapted from Shingu et al.'s approach, requires only \(3w\) register qubits and four ancillary qubits. Although our protocol employs more qubits due to the added verification step, it still uses significantly fewer qubits than the BFK-based protocols.

\section{Conclusion and Discussion}
\label{conclusion}
In this paper, we propose a protocol for cloud-based VQAs that extends the work of Shingu et al. by incorporating verifiability and increasing tolerance to channel loss. We have also demonstrated the protocol's blindness and correctness, ensuring security and accuracy in cloud-based VQAs. Additionally, we compare the client's quantum capabilities and the server's resource consumption with those in Shingu et al.'s protocol and BFK-based protocols, demonstrating that our protocol maintains low quantum resource consumption for the server and minimal quantum capabilities for the client.

The proposed protocol can be further extended through two potential schemes. First, if the client possesses multiple photon detectors, the \(J(\phi)\) operators can be performed on multiple register qubits in parallel, thereby accelerating the protocol's runtime. However, this approach increases the server's quantum resource consumption due to the preparation of additional ancillary qubits. Second, the protocol can be adapted for the client without quantum capabilities by employing the double-server blind quantum computation method~\cite{morimae2013secure}, which is compatible with the proposed protocol.

Further research is necessary in several areas. Reducing communication costs while ensuring security remains a significant challenge. Additionally, exploring alternative verification methods to reduce quantum resource consumption further would be beneficial. Lastly, instead of directly transforming parameterized gate-based quantum circuits into MBQC patterns, adopting an MBQC-native approach~\cite{schroeder2023deterministic,chan2024measurement,calderon2024measurement} for VQAs could offer improved circuit depth reduction, making it suitable for adaptation to cloud-based VQAs. Our protocol has the potential to pave the way for real-world applications of cloud-based quantum computing.

\section*{Acknowledgement}
This work is supported by the National Natural Science Foundation of China under Grant Nos. 62072119 and 62271436.

\appendix

\section{Derivation of \(J(\phi)\) within Universal Gate Patterns}
\label{appendix}

We analyze the operations within the circuit depicted in Fig.~\ref{fig:brickwork}(c), involving three ancillary qubits and one register qubit \(R = \ket{\psi}_R\). Let \(A_1\), \(A_2\), and \(A_3\) represent the ancillary qubits. After the client measures one half of the Bell pair, yielding the result \(s_0\), the ancillary qubits are prepared in the initial states:
\begin{equation}
    A_1 = Z^{s_0} R_Z\left(\frac{k\pi}{4}\right) \ket{+}_{A_1}, \quad A_2 = \ket{+}_{A_2}, \quad A_3 = \ket{0}_{A_3}.
\end{equation}

Let \(s_1\), \(s_2\), and \(s_3\) denote the measurement results of \(A_1\), \(A_2\), and \(A_3\), respectively. The resulting operations on the register and ancillary qubits are:

\begin{equation}
\begin{split}
    &\bigl[H_R \otimes M_{A_3}(\phi')\bigr] CX_{R{A_3}} \bigl[H_R \otimes M_{A_2}(Z)\bigr] CX_{R{A_2}} \bigl[H_R \otimes M_{A_1}(Z)\bigr]  \\
    &\quad CX_{R{A_1}}\left[\ket{\psi}_R \otimes Z^{s_0} R_Z\left(\frac{k\pi}{4}\right) \ket{+}_{A_1} \otimes \ket{+}_{A_2} \otimes \ket{0}_{A_3}\right] \\
    &= \bigl[H_R \otimes M_{A_3}(\phi')\bigr] CX_{R{A_3}} \bigl[H_R \otimes M_{A_2}(Z)\bigr] CX_{R{A_2}} \\
    &\quad  \Bigl\{ \frac{1}{\sqrt{2}}\left[ H_R R_Z\left(\frac{k\pi}{4} + s_0\pi\right)\ket{\psi}_R \otimes \ket{0}_{A_1} \right]\otimes \ket{+}_{A_2} \otimes \ket{0}_{A_3} +   \\
    &\quad \frac{1}{\sqrt{2}}\left[ H_R R_Z\left(-\frac{k\pi}{4} - s_0\pi\right)\ket{\psi}_R \otimes \ket{1}_{A_1} \right]  \otimes \ket{+}_{A_2} \otimes \ket{0}_{A_3}\Bigr\}.
\end{split}
\end{equation}

The measurement result \(s_1\) of \(A_1\) determines the \(Z\)-rotation angle for \(\ket{\psi}_R\). We can simplify the above expression by omitting \(A_1\):

\begin{equation}
\begin{split}
    &\left[H_R \otimes M_{A_3}(\phi')\right] CX_{R{A_3}} \left[H_R \otimes M_{A_2}(Z)\right] CX_{R{A_2}}  \\
    &\quad \left[H_R R_Z\left((-1)^{s_1} \left(\frac{k\pi}{4} + s_0\pi\right)\right)\ket{\psi}_R \otimes \ket{+}_{A_2} \otimes \ket{0}_{A_3}\right] \\
    &= \left[H_R \otimes M_{A_3}(\phi')\right] CX_{R{A_3}} \left[I_R \otimes \left(M_{A_2}(Z) H_{A_2}\right)\right] CZ_{R{A_2}} \\
    &\quad \left[R_Z\left((-1)^{s_1} \left(\frac{k\pi}{4} + s_0\pi\right)\right)\ket{\psi}_R \otimes \ket{0}_{A_2} \otimes \ket{0}_{A_3}\right].
\end{split}
\end{equation}

Since the \(CZ\) operation does not entangle \(\ket{0}_{A_2}\) with the register qubit, we can omit the operations on \(A_2\). Thus, we obtain the following expression for the remaining qubits:

\begin{equation}
\begin{split}
    &\left[H_R \otimes M_{A_3}(\phi')\right] CX_{R{A_3}} \\
    &\quad \left[R_Z\left((-1)^{s_1} \left(\frac{k\pi}{4} + s_0\pi\right)\right)\ket{\psi}_R \otimes \ket{0}_{A_3}\right] \\
    &= \frac{1}{\sqrt{2}} H_R R_Z\left((-1)^{s_1} \left(\frac{k\pi}{4} + s_0\pi\right) - \phi'\right)\ket{\psi}_R \otimes \ket{0}_{A_3} +\\
    &\quad \frac{1}{\sqrt{2}} X H_R R_Z\left((-1)^{s_1} \left(\frac{k\pi}{4} + s_0\pi\right) - \phi'\right)\ket{\psi}_R \otimes \ket{1}_{A_3}.
\end{split}
\end{equation}

The measurement result \(s_3\) of \(A_3\) determines the state of \(R\):
\begin{equation}
    R = X^{s_3} H_R R_Z\left((-1)^{s_1} \left(\frac{k\pi}{4} + s_0\pi\right) - \phi'\right)\ket{\psi}_R.
\end{equation}

Applying Eq.~\ref{eq6}, the above simplifies to:

\begin{equation}
\begin{split}
       R &= X^{s_3} H_R R_Z\left(\phi - r\pi\right)\ket{\psi}_R\\
       &= X^{s_3 \oplus r} H_R R_Z(\phi)\ket{\psi}_R. 
\end{split}
\end{equation}

Thus, \(J(\phi)\) is obtained with Pauli byproducts on the register qubit \(R\).


\begin{thebibliography}{50}

\bibitem{divincenzo1995quantum}DiVincenzo, David P. Quantum computation. {\em Science}. \textbf{270}, 255-261 (1995)

\bibitem{childs2005secure} Childs, Andrew M. {\em Quantum Information \& Computation}. \textbf{5}, 456-466 (2005)

\bibitem{deng2004secure}Deng, Fu-Guo, and Gui Lu Long. {\em Physical Review A}. \textbf{69}, 052319 (2004)

\bibitem{broadbent2009universal}Broadbent, Anne, Joseph Fitzsimons, and Elham Kashefi. {\em 2009 50th Annual IEEE Symposium On Foundations Of Computer Science}. pp. 517-526 (2009)


\bibitem{briegel2009measurement}Briegel, Hans J and Browne, David E and D{\"u}r, Wolfgang and Raussendorf, Robert and Van den Nest, Maarten. {\em Nature Physics}. \textbf{5}, 19-26 (2009)

\bibitem{morimae2014verification}Morimae, Tomoyuki. {\em Physical Review A}. \textbf{89}, 060302 (2014)
\bibitem{fitzsimons2017unconditionally}Fitzsimons, Joseph F and Kashefi, Elham.  {\em Physical Review A}. \textbf{96}, 012303 (2017)

\bibitem{barz2013experimental}Barz, Stefanie and Fitzsimons, Joseph F and Kashefi, Elham and Walther, Philip. {\em Nature Physics}. \textbf{9}, 727-731 (2013)

\bibitem{morimae2013blind}Morimae, Tomoyuki and Fujii, Keisuke. {\em Physical Review A}. \textbf{87}, 050301 (2013)

\bibitem{li2021blind}Li, Qin and Liu, Chengdong and Peng, Yu and Yu, Fang and Zhang, Cai {\em Optics \& Laser Technology}. \textbf{142} pp. 107190 (2021)

\bibitem{morimae2013secure}Morimae, Tomoyuki and Fujii, Keisuke. {\em Physical Review Letters}. \textbf{111}, 020502 (2013)

\bibitem{li2014triple}Li, Qin and Chan, Wai Hong and Wu, Chunhui and Wen, Zhonghua.  {\em Physical Review A}. \textbf{89}, 040302 (2014)

\bibitem{xu2021universal}Xu, Hai-Ru and Wang, Bang-Hai.  {\em Laser Physics Letters}. \textbf{19}, 015202 (2021)

\bibitem{zhang2022measurement}Zhang, Xiaoqian.  {\em Quantum Information Processing}. \textbf{21}, 14 (2022)

\bibitem{cao2023multi}Cao, Shuxiang.  {\em New Journal Of Physics}. \textbf{25}, 103028 (2023)

\bibitem{sciarrino2023multi}Sciarrino, Fabio and Polacchi, Beatrice and Leichtle, Dominik and Limongi, Leonardo and Carvacho, Gonzalo and Milani, Giorgio and Spagnolo, Nicolo and Kaplan, Marc and Kashefi, Elham.  {\em Nature Communications}. (2023)

\bibitem{das2022blind}Das, Aritra and Sanders, Barry C.  {\em Physical Review A}. \textbf{106}, 012421 (2022)

\bibitem{gustiani2021blind}Gustiani, Cica and DiVincenzo, David P.  {\em Physical Review A}. \textbf{104}, 062422 (2021)

\bibitem{barz2012demonstration}Barz, Stefanie and Kashefi, Elham and Broadbent, Anne, and Fitzsimons, Joseph F and Zeilinger, Anton and Walther, Philip. {\em Science}. \textbf{335}, 303-308 (2012)

\bibitem{huang2017experimental}Huang, He-Liang and Zhao, Qi and Ma, Xiongfeng, and Liu, Chang and Su, Zu-En and Wang, Xi-Lin and Li, Li and Liu, Nai-Le and Sanders, Barry C. and Lu, Chao-Yang and Pan, Jian-Wei.  {\em Physical Review Letters}. \textbf{119}, 050503 (2017)

\bibitem{cerezo2021variational}Cerezo, Marco and Arrasmith, Andrew and Babbush, Ryan and Benjamin, Simon C and Endo, Suguru and Fujii, Keisuke and McClean, Jarrod R and Mitarai, Kosuke and Yuan, Xiao and Cincio, Lukasz and others  {\em Nature Reviews Physics}. \textbf{3}, 625-644 (2021)

\bibitem{biamonte2017quantum}Biamonte, Jacob and Wittek, Peter and Pancotti, Nicola and Rebentrost, Patrick and Wiebe, Nathan and Lloyd, Seth.  {\em Nature}. \textbf{549}, 195-202 (2017)

\bibitem{khan2020machine}Khan, Tariq M and Robles-Kelly, Antonio.  {\em IEEE Access}. \textbf{8} pp. 219275-219294 (2020)

\bibitem{chen2021federated}Chen, Samuel Yen-Chi and Yoo, Shinjae.  {\em Entropy}. \textbf{23}, 460 (2021)

\bibitem{PhysRevLett.113.130503}Rebentrost, Patrick and Mohseni, Masoud and Lloyd, Seth.  {\em Physical Review Letters}. \textbf{113}, 130503 (2014,9)

\bibitem{dong2008quantum}Dong, Daoyi and Chen, Chunlin and Li, Hanxiong and Tarn, Tzyh-Jong.  {\em IEEE Transactions On Systems, Man, And Cybernetics, Part B (Cybernetics)}. \textbf{38}, 1207-1220 (2008)

\bibitem{li2021quantum}Li, Weikang and Lu, Sirui and Deng, Dong-Ling.  {\em Science China Physics, Mechanics \& Astronomy}. \textbf{64}, 100312 (2021)

\bibitem{wang2022delegated}Wang, Yuxun and Quan, Junyu and Li, Qin.  {\em 2022 14th International Conference On Wireless Communications And Signal Processing (WCSP)}. pp. 804-808 (2022)

\bibitem{houshmand2018minimal}Houshmand, Monireh and Houshmand, Mahboobeh and Fitzsimons, Joseph F.  {\em Physical Review A}. \textbf{98}, 012318 (2018)

\bibitem{shingu2022variational}Shingu, Yuta and Takeuchi, Yuki and Endo, Suguru, and Kawabata, Shiro and Watabe, Shohei and Nikuni, Tetsuro and Hakoshima, Hideaki and Matsuzaki, Yuichiro.  {\em Physical Review A}. \textbf{105}, 022603 (2022)
\bibitem{anders2010ancilla}Anders, Janet and Oi, Daniel K. L. and Kashefi, Elham and Browne, Dan E. and Andersson, Erika.  {\em Physical Review A}. \textbf{82}, 020301 (2010)


\bibitem{agarwal2012quantum}Agarwal, Girish S. Quantum optics. (Cambridge University Press,2012)

\bibitem{nakaji2022approximate}Nakaji, Kouhei and Uno, Shumpei and Suzuki, Yohichi, and Raymond, Rudy and Onodera, Tamiya and Tanaka, Tomoki and Tezuka, Hiroyuki and Mitsuda, Naoki and Yamamoto, Naoki. {\em Physical Review Research}. \textbf{4}, 023136 (2022)

\bibitem{jeswal2019recent}Jeswal, SK and Chakraverty, S. {\em Archives Of Computational Methods In Engineering}. \textbf{26} pp. 793-807 (2019)

\bibitem{KingBa15}Kingma, Diederik and Ba, Jimmy.  {\em International Conference On Learning Representations (ICLR)}. (2015)


\bibitem{schuld2019evaluating}Schuld, Maria and Bergholm, Ville and Gogolin, Christian and Izaac, Josh and Killoran, Nathan. {\em Physical Review A}. \textbf{99}, 032331 (2019)

\bibitem{gothen2023euler}Gothen, Peter and Guedes de Oliveira, Ant{\'o}nio.  {\em The College Mathematics Journal}. \textbf{54}, 171-175 (2023)


\bibitem{kuzanyan2018three}Kuzanyan, AS and Kuzanyan, AA and Nikoghosyan, VR.  {\em Journal Of Contemporary Physics (Armenian Academy Of Sciences)}. \textbf{53} pp. 338-350 (2018)


\bibitem{wenner2014catching}Wenner, J and Yin, Yi and Chen, Yu and Barends, R and Chiaro, B and Jeffrey, E and Kelly, J and Megrant, A and Mutus, JY and Neill, C and others. \&  {\em Physical Review Letters}. \textbf{112}, 210501 (2014)

\bibitem{sueki2013ancilla}Sueki, Takahiro and Koshiba, Takeshi and Morimae, Tomoyuki. {\em Physical Review A}. \textbf{87}, 060301 (2013)

\bibitem{schroeder2023deterministic}Schroeder, Anna and Heller, Matthias and Gachechiladze, Mariami.  {\em New Journal Of Physics}. (2023)


\bibitem{chan2024measurement}Chan, Albie, and Shi, Zheng and Dellantonio, Luca and D\"ur, Wolfgang and Muschik, Christine A. {\em Physical Review Letters}. \textbf{132}, 240601 (2024)

\bibitem{calderon2024measurement}Calder{\'o}n, Luis Mantilla and Feldmann, Polina and Raussendorf, Robert and Bondarenko, Dmytro.  {\em ArXiv Preprint ArXiv:2405.08319}. (2024)
\end{thebibliography}
\end{document}